\documentclass[12pt]{article}
\usepackage{amsmath}
\usepackage{amsthm}
\usepackage{import}
\usepackage{bbm}
\usepackage{tikz}
\usepackage{geometry}
\usepackage{comment}
\usepackage[onehalfspacing]{setspace}  
\usepackage[title]{appendix}
\usepackage[utopia]{mathdesign}

\usepackage{sectsty}
\sectionfont{\centering \Large \scshape }
\subsectionfont{\centering \large  \scshape}
\usepackage[american]{babel}
\usepackage{csquotes}
\usepackage[style=apa, uniquelist=false, backend=biber, doi=false,isbn=false,url=false]{biblatex}
\DeclareLanguageMapping{american}{american-apa}
\defbibenvironment{bibliography}
{\enumerate{}
{\setlength{\leftmargin}{\bibhang}%
\setlength{\itemindent}{-\leftmargin}%
\setlength{\itemsep}{\bibitemsep}%
\setlength{\parsep}{\bibparsep}}}
{\endenumerate}
{\item}
\usepackage[colorlinks,citecolor=blue]{hyperref}

\usetikzlibrary{shapes.geometric}
\usetikzlibrary{positioning}


\newcounter{parentnumber}
\theoremstyle{plain}
\newtheorem {theorem}{Theorem}
\newtheorem{proposition}{Proposition}

\newtheorem {conjecture}{Conjecture}

\newtheorem {lemma}{Lemma}
\theoremstyle{definition}
\newtheorem {example}{Example}

\usepackage{color}

\newcommand{\N}{\mathbb{N}}
\newcommand{\R}{\mathbb{R}}

\title{\vspace{-4em}A Group Public Goods Game with Position Uncertainty\thanks{We wish to thank conference participants at SING17, where this work has been presented, and particularly, Trivikram Dokka, Konstantinos Georgalos, Alexander Matros and Jaideep Roy for their helpful comments.}\\}
\author{ \normalsize Chowdhury Mohammad Sakib Anwar\thanks{BLDT, University of Winchester. Email: Sakib.Anwar@winchester.ac.uk}
\and \normalsize  Jorge Bruno\thanks{BLDT, University of Winchester. Email: Jorge.Bruno@winchester.ac.uk}
\and \normalsize  Sonali SenGupta\thanks{Queens Management School, Queen's University Belfast. Email: S.SenGupta@qub.ac.uk}}
\date{\normalsize \today}
\begin{document} 
\maketitle
\begin{abstract}
    We model a dynamic public good contribution game, where players are (naturally) formed into groups. The groups are exogenously placed in a sequence, with limited information available to players about their groups' position in the sequence. Contribution decisions are made by players simultaneously and independently, and the groups' total contribution is made sequentially. We try to capture both inter and intra-group behaviors and analyze different situations where players observe partial history about total contributions of their predecessor groups. Given this framework, we show that even when players observe a history of defection (no contribution), a cooperative outcome is achievable. This is particularly interesting in the situation when players observe only their immediate predecessor groups' contribution, where we observe that players play an important role in motivating others to contribute. 
    \medskip
\begin{flushleft}\textbf{Keywords} : Social Dilemmas, Public Goods, Position Uncertainty, Voluntary Contributions, Fundraising, Groups.\end{flushleft}\par
\begin{flushleft}\textbf{JEL Codes}: C72, D82, H41 \ \end{flushleft}\par
\end{abstract}

\newpage
\section{Introduction}

In public good provisions, the non-cooperative interplay of the players typically results in lower levels of provisions, compared to the socially optimal levels. Free-riding is a major stumbling block for efficient public good provision. Extensions of the standard one-shot public good provision game in the form of repeated strategic interactions over an infinite time horizon, or non-standard preferences, etc.,  may help converge to the socially optimal levels by improving contribution by self-interested agents.\footnote{See \textcite{friedman1971non, dal2010institutions, duffy2009cooperative}, among others, for theoretical and experimental work on repeated interactions, and \textcite{andreoni1990impure, fehr2000Gachter, fehr2002altruistic} on non-standard preferences like altruism, warm-glow effects, etc.}

In this paper we model a dynamic public goods game where a finite number of players are divided into groups, and each group is exogenously placed in a sequence, with limited information available to the players about their group's position in the sequence. Members of each group first decide, simultaneously and independently, whether or not to contribute towards the public good and the groups' total contribution is made sequentially. We assume that players observe some of their immediate predecessor groups' total contribution. Given this framework, we show that a cooperative outcome (where self-interested players contribute) is achievable even with a history of defection.

To help fix ideas and motivate our setting, consider a crowd-funding project which relies on raising money from a large number of individuals who make (small) contributions typically via an online platform. Typical of such a situation is that contributions are updated and visible with a time lag leading to a temporal aspect to contributor behavior, that is, between updates all potential contributors act on whether and how much to contribute simultaneously and independently. Naturally, all contributors that act within the same time interval can be viewed as a group in our setting. For confidentiality reasons, a platform may wish to not reveal individual contributions at any time to potential future contributors but would reveal immediate predecessor group or groups' contribution as a means to motivate more contributions.

Our model generalizes the model of \textcite{gallice2019co} public goods game with sequential uncertainty and observational learning, by capturing natural group formation (as illustrated via crowd-funding scenario) and capturing both inter and intra-group behaviors into account. We first consider symmetric (same size) groups, and then generalize our analysis to discuss asymmetric groups. Note that the symmetric groups appears as a special case of asymmetric groups, and all the analysis and results follow. Also note that in our model, individuals in a group only observe partial history of their immediate groups' total contributions, and do not observe individual players' contributions.\footnote{See also \textcite{monzon2019observational} for a related study where players are ordered in a sequence and they receive private signals about an uncertain state of the world and observe some history of immediate predecessors previous actions. Another interesting work of positional uncertainty is due to \textcite{monzon2014observational}.}

We characterize the  sequential equilibrium  for the following cases: (i) when players observe at least two of their predecessor  groups' total contribution, and (ii) when players observe only their immediate predecessor group's total contribution.  We find full contribution can occur in the first case  if the rate of return from the public goods is high enough. In the second case we show contributions can occur in two concurrent types of equilibria: pure and mixed. If we assume group sizes of 1 then our results reduce to that of \textcite{gallice2019co}, where exactly one mixed strategy exists. However, if at least one group size is larger than one we find at least two mixed strategy equilibrium. We illustrate this as an instance of players becoming pivotal and inducing other players to contribute - even though subsequent groups might still witness defections within a given player's group. In particular, our result on pivotal agent contributes towards the literature on pivotal voter model.\footnote{See \textcite{ledyard1984pure,palfrey1983strategic,palfrey1985voter}. In these models voter's main concern is to affect the outcome of the election, and they are not concerned about motivations like warm-glow and altruism. Agents vote if the expected benefit voting is larger the cost.} Our analysis (Theorem \ref{thm:mixedStrat}) shows that a player will be pivotal in encouraging others to contribute even if they witness a sample of defection (no contribution in history). In fact, we show that an agent's pivotality has a unique local maxima as a function of the probability of contribution in a mixed strategy scenario; when a player decides to contribute, succeeding groups will observe a smaller number of defections and, therefore, increase their likely-hood to contribute.  This influence reaches a maximum and then start to decline as her contribution becomes less influential in inducing others to contribute.

We mainly contribute towards the literature on sequential (or dynamic) public goods contribution games, where previous work studies how the timing of contribution can affect the total amount raised through contribution.  \textcite{Varian1994} finds that sequential contributions lowers total contribution.\footnote{\textcite{admati1991joint} predicts an inefficient allocation of resources even for a contributions mechanism that uses sequential contributions. The paper also shows that the unique subgame-perfect equilibrium in a sequential threshold public goods game with full refund leads to an efficient outcome. See \textcite{bag2011sequential} for an extension of this work.} Experimental evidence suggests that sequential mechanism might perform better in raising funds compared to simultaneous contribution mechanism \parencite{Andreoni2002,gachter2010sequential}.\footnote{Several papers focus on endogenously determining the public goods game to play. See for example \textcite{potters2005after,romano2001charities,vesterlund2003informational,bracha2011seeds}.}
\textcite{figuieres2012vanishing} tests the effect of information about contribution levels in a sequential public goods game and show that the level of contributions increases when subjects are informed about the contributions of the predecessors. In a more recent work, \textcite{tajika2020contribute} models a dynamic public goods game where players are allowed to make a one-off discrete contribution. Each player decides the period and the amount to contribute, and has full information about the total contributions prior to making their decision. We also contribute to the literature on games with position uncertainty and observational learning. Previous work on position uncertainly takes a different approach than us, where they typically consider scenarios where a principal can choose the history of past actions to show to the agents \parencite{nishihara1997resolution}\footnote{We adopt a different approach from \textcite{gershkov2009optimal,doval2020sequential}. For example, \textcite{gershkov2009optimal} models optimal voting schemes where the principal chooses a protocol to induce voters to acquire costly information and reveal it truthfully. See, for example, \textcite{banerjee1992simple}, \textcite{ccelen2004observational}, \textcite{hendricks2012observational}, \textcite{guarino2013social}, \textcite{garcia2018consumer} on papers related to observational learning.}.  We follow a different approach where agents learn directly by observing their immediate predecessor groups' contributions.

The outline of the paper is as follows: Section 2 describes our model. The main results are discussed in Section 3 where we first treat the case of symmetric groups and then extend our results to asymmetric ones. Unless otherwise stated, all the proofs are in Appendix.

\section{The Model for Groups}
We adapt the notation developed by \textcite{gallice2019co} to suit the group-based nature of our results. Let $I=\{1,2,...,N\}$ be a set of players and consider a game with $b\leq N$ many groups composed of players from $I$. More precisely, let $q: I \to [b]$ be an onto function, where $[b]$ denotes the set of the first $b$ positive integers, and $\mathcal{Q}$ be the set of all such functions. We assume that all said functions are equally likely: $Pr(Q=q) = \sum_{i=0}^{b-1}(b-i)^{N}{b\choose i}$ for all $q\in \mathcal{Q}$, where $Q$ is a random variable. In turn, Group $t$ becomes the collection of players $j\in I$ for which $Q(j) = t$.

The timing of the game is as follows. First,  Nature (a non-strategic player) chooses the order of the sequence $q$. Players do not observe the choice of Nature. Players then observe a sample of past actions through a simple sampling rule (which we formally describe below). Players in each group must choose one of two actions simultaneously and independently. More precisely, Player $i$ in Group $t$ must choose one of two actions $a_{i,t} \in \{C,D\}$: action $a_{i,t}=C$ implies a contribution of  one unit while $a_{i,t}=D$ implies no contribution. The total group contribution goes towards the common fund. The common fund is then redistributed to all players and we assume $r$ as the rate of return. We adopt the standard notation $G_{-i}= \sum_{j\neq i} \mathbbm{1}\{a_{j,t}=C\}$ to denote the number of players (other than $i$) who contribute. Payoffs $u_i(a_i,G_{-i})$ of player $i$ is as follows
\[
    u_i(C,G_{-i}) = \frac{r}{N}(G_{-i}+1)-1
\]
\[
    u_i(D,G_{-i})=\frac{r}{N}(G_{-i}),
\]
where $r$ is the return from contributions, and $\frac{r}{N}$ gives the marginal per capita return from the public goods. Whilst for a fixed $G_{-i}$ it is evident that $u_i(C,G_{-i}) \leq u_i(D,G_{-i})$, it is the effect of Player $i$'s contribution, or lack thereof, on subsequent players that we study here. Hence, the value of $G_{-i}$ will often change for different values of $a_{i,t}$.

For $t\leq b$, the symbol $A_t = (a_{i , t})$ denotes actions of the players in Group $t$ and $h_t=(A_t)_{t=1}^{t-1}$ denotes a possible history of actions up to Group $t-1$. Let $H_t$ be the random history at period $t$ with realizations $h_t\in \mathcal{H}_t$ and let $\mathcal{H}_1=\{\emptyset \}$.\footnote{We use period and position interchangeably throughout the paper, as they imply the same in our context.} Players play an extensive form game with imperfect information where a player is given a sample $\zeta$ of the actions of the $m\geq 1$ preceding groups. The value of $m$ is known {\it a priori}. That is, players observe a sample $\zeta=(\zeta',\zeta'')$, where $\zeta'$ states the number of groups sampled and $\zeta''$ states the number of contributors in that sample. Of course, a player in Group $t<m$ will be provided with a sample $\zeta=(t-1,\zeta'')$. In turn, players in the first group observe $\zeta_1=(0,0)$ and players in groups positioned between 2 to $m$ observe the actions of all their predecessor groups, and so they can infer their exact position in the sequence from the sample they receive. Formally,  letting $g_{t}=\sum_{Q(i)= t}\mathbbm{1}\{a_{i,t}=C\} $ denote the total contributions in Group $t$, players in Group $t$ receive a sample 
$\zeta_t: \mathcal{H}_t \to \mathcal{S}=\N^2$ containing a tuple:

$$ \zeta_t(h_t)= \Big( \underbrace{\min\{m,t-1\}}_{=\zeta'},\quad \underbrace{\sum_{k=\max\{1,m-t\}}^{t-1}g_{k}}_{=\zeta''}\Big)$$

We use \textcite{kreps1982sequential} sequential equilibrium. Player $i$'s strategy is a function $\sigma_i(C|\zeta): \mathcal{S}\to [0,1]$ that specifies the probability of contributing given the sample received. Let $\sigma=\{\sigma_i\}_{i\in I}$ denote a strategy profile and $\mu=\{\mu_i\}_{i\in I}$ a system of beliefs. A pair $(\sigma, \mu)$ represents an \textit{assessment}. Assessment $(\sigma ^*, \mu^*)$ is a {\it sequential equilibrium} if $\sigma^*$ is sequentially rational given $\mu^*$, and $\mu^*$ is consistent given $\sigma^*$. Let $\mathcal{H}= \cup_{t=1}^{n}\mathcal{H}_t$ be the set of all possible histories. Given a profile of play $\sigma$ let $\mu_i$ denote Player $i$'s beliefs about the history of play : $\mu_i(h|\zeta) : \mathcal{H} \times \mathcal{S}\to [0,1]$, with $\sum_{h \in \mathcal{H}} \mu_i(h|\zeta) =1$ for all $\zeta \in \mathcal{S}$.

\section{Results}

We first treat the case where all groups are of the same size $n = \frac{N}{b}$ and later extend our results to the asymmetric case in Section~\ref{sec:asym}.

\subsection{Symmetric groups with sample size $m>1$\label{sec:sym>1}}

Assume a sample size of $m>1$. Given a sample $\zeta = (\zeta',\zeta'')$ with $m \geq \zeta' > 1$ we demonstrate that the simple strategy of ``contributing unless a defection is observed" yields a sequential equilibrium provided that $r$ is {\it large enough.} Since the proof of the following is simple and self-contained we present it here. For completeness, in what follows we let $\sigma_i^k$ denote a sequence of strategies with $\sigma_i^k(C\mid \zeta)= 1 - (s_k)$ and $\sigma_i^k(D\mid \zeta) = (s_k)$ where $(s_k)$ is any non-trivial real null sequence, and put $\mu_i^k$ as the induced belief for strategy $\sigma_i^k$ for each $k\in \N$.
\begin{proposition}

\label{lem:m>1Symmetric} Consider the profile of play

\begin{align*}
    \sigma_i^*(C\mid \zeta) =
\begin{cases}
1, & \text{if $\zeta$ contains no defections} \\
0, & \text{otherwise.}
\end{cases}
\end{align*}
It follows that $(\sigma^*,\mu^*)$ is a sequential equilibrium provided that $r \geq \frac{2N}{2N - (b+m-1)n}$.
\end{proposition}
\begin{proof}
Consider a player in Group $t$ and a history $\zeta = (m,c)$ with $mn>c$. For groups of size at least $2$, an agent that observes $\zeta$ is aware that every other agent in her group also witnesses a defection and given the pure profile of play, the effect of the defection would extend beyond her group regardless of her contribution, or lack thereof. Even when $n=1$ players in subsequent groups will still witness defections as $m>1$ and $c+1 + kn = mn$ for some $k\in \N$. This means that if the defection occurs in Group $t'<t$ and all other players in groups succeeding $t'$ must defect, then a defection must inevitably appear in Group $t-1$. It follows that defection after defection is optimum given any value of $r$. \\

In contrast to the previous case, if $\zeta = (\zeta',\zeta'n)$ it is simple to deduce that a player in Group $t>m$ will contribute precisely when 
\[
\frac{r}{N}N -1 \geq \frac{r}{N}\frac{(N+m-1)b}{2}
\]
since she expects her position to be in the mid-point between $m+1$ and $b$. In turn,
\[
r \geq \frac{2N}{2N - (b+m-1)n}.
\]
That is, perpetuating contribution only makes sense when $r$ is at least as large as $\frac{2N}{2N - (b+m-1)n}$.
 Any other player in a Group $t<m$ will have an even greater incentive to contribute as their gains from contributing will be larger than those in a Group $t'>m$.
\end{proof} 

\subsection{Symmetric groups with sample size $m=1$\label{sec:sym=1}}

In this section we assume that $m=1<b$. When $n=1$, \textcite{gallice2019co} prove that a pure strategy equilibrium exists for all $r\in [2,3-\frac{3}{b+1}]$. In contrast, we show that for the scenario with $n>1$, a pure strategy equilibrium exists for all values of $r\geq 2$. The main difference with the $n=1$ and the $n>1$ case lies in the inability of an agent in the latter that observes a defection to prevent further defections since she is aware that every other agent in her group also witnesses said defection; much like the case of $m>1$ with defections, as treated in the previous section. It follows, as illustrated in the Appendix (Section \ref{apendix:Proofs}), that a pure strategy Nash equilibrium exists for all such values of $r$.
\begin{theorem}[Pure Strategies with $m=1<n$]\label{thm:pureStrat} For any value of $r\geq 2$ and given the profile of play
\begin{align*}
    \sigma_i^*(C\mid \zeta) =
\begin{cases}
1, & \zeta \in\{(0,0), (1,n)\}\\
0, & \text{otherwise}
\end{cases}
\end{align*}
for all $i\in I$, the assessment $(\sigma^*, \mu^*)$ is a sequential equilibrium,
\end{theorem}

Continuing with the case $m=1<n$, we demonstrate that - concurrent to the equilibrium from Theorem~\ref{thm:pureStrat} - for large enough values of $r$ there are mixed strategy equilibria where agents forgive defections with probability $\gamma\in (0,1)$. In fact, surprisingly enough, for all but the smallest such value of $r$ there exist {\it at least two} mixed strategy equilibria.

\begin{theorem}[Mixed Strategies with $m=1<n$]\label{thm:mixedStrat} There exists a value $r^\sharp<N$ with a corresponding $\gamma^\sharp \in (0,1)$ where, for all $i\in I$, the profile of play
    \begin{align*}
    \sigma_i^*(C\mid \zeta) =
\begin{cases}
1, & \zeta\in\{(0,0), (1,n)\}\\
\gamma^\sharp, & \text{otherwise}
\end{cases}
\end{align*}
yields a sequential equilibrium $(\sigma^*, \mu^*)$. Moreover, for all values $r>r^\sharp$ there exist two distinct values $\gamma^1_r,\gamma^2_r \in(0,1)$ so that, for all $i\in I$, the profiles of play

\begin{align*}
    \sigma_{i,1}^*(C\mid \zeta) =
\begin{cases}
1, & \zeta \in\{(0,0), (1,n)\}\\
\gamma^1_r, & \text{otherwise},
\end{cases}
\end{align*}

\begin{align*}
    \sigma_{i,2}^*(C\mid \zeta) =
\begin{cases}
1, & \zeta \in\{(0,0), (1,n)\}\\
\gamma^2_r, & \text{otherwise}
\end{cases}
\end{align*}
 establish two distinct sequential equilibria $(\sigma^*_1, \mu^*_1)$ and $(\sigma^*_1, \mu^*_2)$, respectively.

\end{theorem}
 
 To understand the above results consider the following standard strategy given a profile $\zeta$
\begin{align*}
    \sigma^*(C\mid \zeta) =
\begin{cases}
1, & \zeta \in\{(0,0), (1,n)\}\\
\gamma, & \text{otherwise}.
\end{cases}
\end{align*}
where $\gamma \in(0,1)$, and for a profile $\overline{\zeta}$ - with a defection - set for a Player $j$ in Group $t$ the strategies
\[
\sigma_j^C(\zeta)
=
\begin{cases}
\sigma^*_j(\zeta), & \zeta \not = \overline{\zeta}\\
1, & \zeta = \overline{\zeta}.
\end{cases}
\]

\[
\sigma_j^D(\zeta)
=
\begin{cases}
\sigma^*_j(\zeta), & \zeta \not = \overline{\zeta}\\
0, & \zeta = \overline{\zeta}.
\end{cases}
\]

\noindent 
Put $\phi_t(\gamma)$ and $\psi_t(\gamma)$ as the number of additional contributions Player $j$ expects from contributing rather than defecting, and the likelihood  of being in Group $t$ after observing a defection, respectively. The function
 \[
H(\gamma) = \frac{r}{N}\sum_{t=2}^b\psi_t(\gamma)\phi_t(\gamma) - 1
\]
then describes the difference in utility between contributing and defecting upon witnessing a defection. Mathematically, as illustrated in Appendix (Section \ref{apendix:Proofs}), we have that for all $\gamma\in (0,1)$: $H(\gamma) > H(0) = H(1) = \frac{r}{N} -1<0$. It follows, by Rolle's Theorem, that $H(\gamma)$ must attain at least one local maximum between $(0,1)$. What is key in proving Theorem~\ref{thm:mixedStrat} is that for at least one $r^\sharp<N$ - and thus all other $r>r^\sharp$ -  $H(\gamma)$'s local maximum is positive. In which case, $H(\gamma)$ exhibits at least two roots and, thus, two distinct sequential equilibria as described in the Theorem. 
Intuitively, we have that as $\gamma$ increases Player $i$ becomes more pivotal in inducing others to contribute. When $H(\gamma)$ reaches its peak the player's pivotality reaches its maximum, but it then starts to decline as her contribution becomes less influential in inducing others to contribute. So we see that $H$ decreases and eventually reaches 0 again. Player $i$ is only pivotal when $\gamma \in [\gamma^*, \gamma ^\#]$ (see Figure~\ref{fig:1}). In Figure~\ref{fig:1} we can also see that when group size is exactly 1, we get the same result as \textcite{gallice2019co}.  
\begin{figure}[htbp]
    \centering
    \includegraphics[scale=0.75]{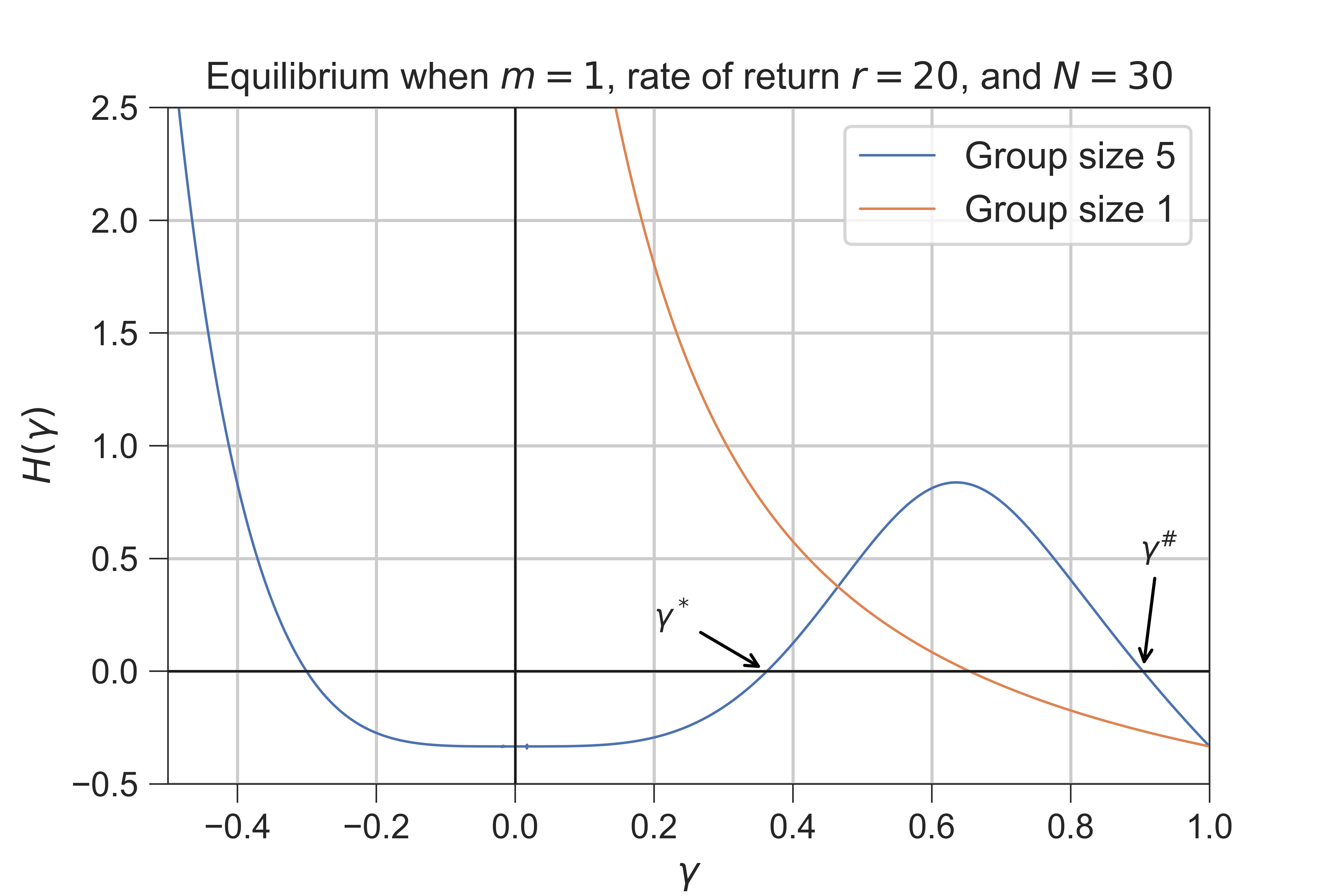}
    \caption{Mixed Strategy Equilibrium}
    \label{fig:1}
\end{figure}

Figure \ref{fig:2} shows that as group size becomes smaller the $H(\gamma)$ peaks become flatter and the equilibrium points move to the left, implying that as the group size reduces, each individual player becomes more influential in persuading others to contribute.

\begin{figure}[htbp]
    \centering
    \includegraphics[scale=0.75]{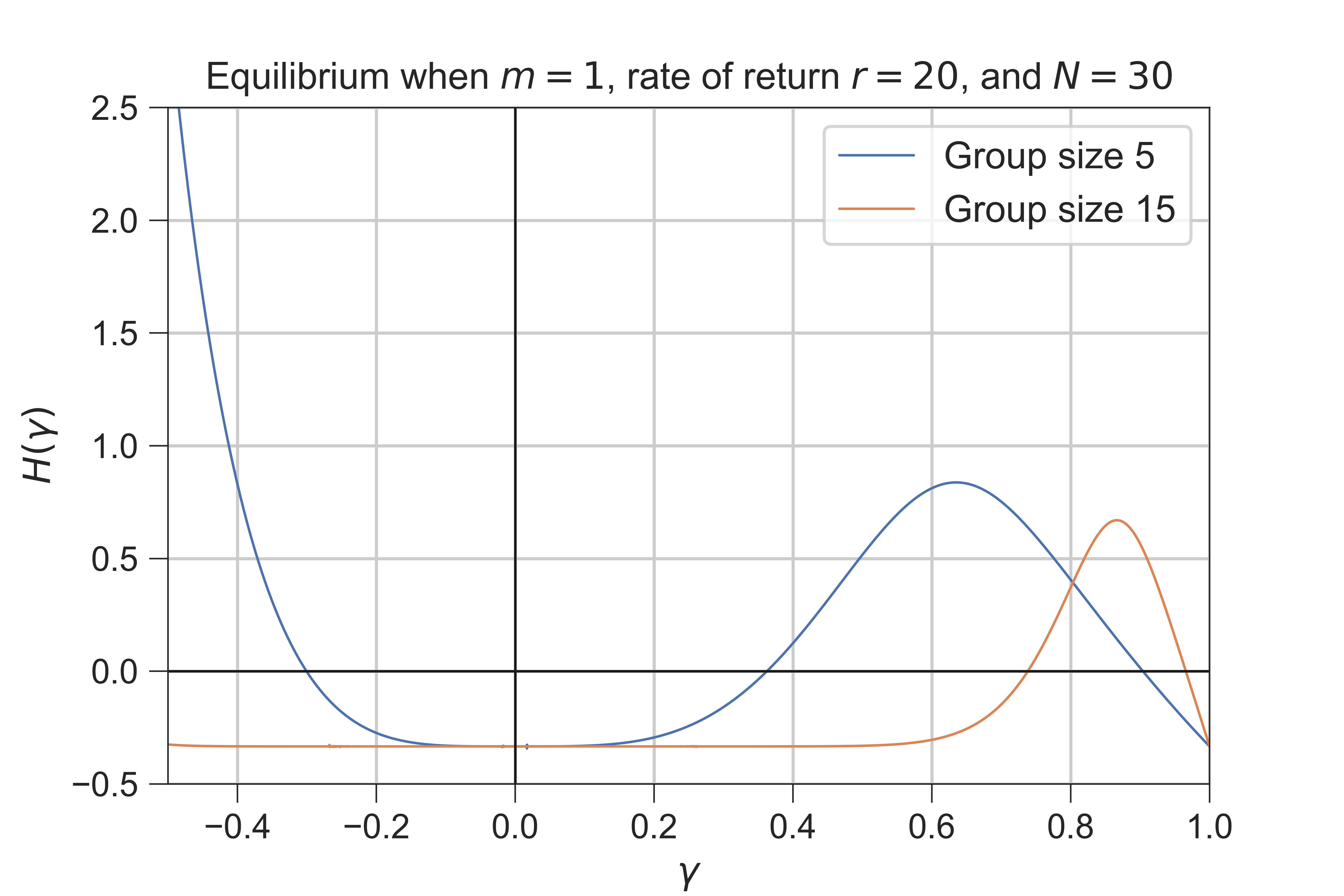}
    \caption{Mixed Strategy Equilibrium-Comparing group sizes}
    \label{fig:2}
\end{figure}

\subsection{Asymmetric groups}\label{sec:asym}
We are now in a position to discuss asymmetric groups: consider $b$ many groups where each group has $n_i$ many members with $\sum_{i=1}^b n_i = N$. Although the results in this section subsume those from Sections~\ref{sec:sym>1} and ~\ref{sec:sym=1} the insights gained from the latter are essential in developing the asymmetric case. In particular, the asymmetric scenario yields similar results to those in the symmetric case.

\begin{proposition}[Asymmetric groups for $m>1$]  Consider the profile of play

\begin{align*}
    \sigma_i^*(C\mid \zeta) =
\begin{cases}
1, & \text{ if $\zeta$ contains no defections} \\
0, & \text{otherwise.}
\end{cases}
\end{align*}
It follows that $(\sigma^*,\mu^*)$ is a sequential equilibrium provided that $
r \geq \frac{2N}{2N-(b+m-1)n}.
$ where $n = \max\{n_1, \ldots, n_k\}$. If the groups sizes $n_1, \ldots, n_b$ are chosen at random then a tighter bound can be deduced by letting $n = \frac{n_1+ \ldots + n_b}{b}$.
\end{proposition}
\begin{proof} We omit the details of this proof as the arguments employed in the proof of Lemma~\ref{lem:m>1Symmetric} can be applied here. 
\end{proof}
It is worth noting that the lower bounds on $r$ above do not represent a tight bound. In fact, it is merely a {\it worst case scenario value for $r$}. Indeed, consider the following example.
\begin{example} Consider the scenario with 3 groups of size $1 = n_1$ and $n_2=n_3 = 2$ and $m=2$. A simple computation yields that the minimum value of $r$ for which players would contribute is $\frac{5}{9}$ whereas the one from the above Lemma is
\[
r = \frac{10}{10 - (3+2-1)2 } = 5. 
\]
If the groups sizes $n_1, \ldots, n_b$ are chosen at random then $r \approx 3$.
\end{example}

Next we focus on the scenario where $m=1$. In the following result we tacitly assume that a member of Group $t>1$ must have knowledge of $n_{t-1}$; this is in order to be able to detect defections. Although the algebra in the asymmetric case becomes more complicated, the intuition for the existence of equilbira (i.e., pivotality) still applies.

\begin{theorem}[Asymmetric groups for $m=1$]\label{thm:mixedStratAsym} Theorem~\ref{thm:pureStrat} and~\ref{thm:mixedStrat} hold for the asymmetric scenario.
\end{theorem}

The proof of Theorem \ref{thm:mixedStratAsym} follows a similar approach to Theorems~\ref{thm:pureStrat} and~\ref{thm:mixedStrat}, and we therefore postpone the proof  to the Appendix (Section \ref{apendix:Proofs}). Moreover, it is important to note that the results derived for asymmetric groups are in the same vein as the symmetric case, and therefore similar discussions follow.

\section{Remarks}

Our main objective behind this project has been to investigate whether better contribution levels are achievable for public goods, if we allow for natural group formation followed by sequential contributions by the respective groups. We draw inspiration from the idea of crowd-funding projects. We make the following main conclusions: (i) full contribution is achievable when players observe at least two of their predecessor groups' contributions, (ii) a player will be pivotal in inducing other players to contribute even with a history of defection, in the situation when players observe only their immediate predecessor groups' total contribution (Theorem \ref{thm:mixedStrat}).    
The reader familiar with Theorem~\ref{thm:mixedStrat} and its proof would probably be inclined to revisit the results from Section~\ref{sec:sym>1} and wonder if a similar result is also possible when $m>1$. That is, is there a \textit{mixed strategy equilibrium} for the case where $m>1$? After all, the situation encountered by player a in Group $t>0$ for $m>1$ is similar to that for a player in any Group $t'>0$ with $m=1<n$; said player is aware that receiving a sample with a defection implies that subsequent groups will also witness a defection. We believe that the arguments developed for Theorem~\ref{thm:mixedStrat} should also apply when $m=1$ and any value of $n$.

\begin{conjecture} [Mixed Strategies with $m=1$] There exists a value $r^\sharp<N$ with a corresponding $\gamma^\sharp \in [0,1)$ where, for all $i\in I$, the profile of play
    \begin{align*}
    \sigma_i^*(C\mid \zeta) =
\begin{cases}
1, & \zeta\in\{(0,0), (1,n)\}\\
\gamma^\sharp, & \text{otherwise}
\end{cases}
\end{align*}
yields a sequential equilibrium $(\sigma^*, \mu^*)$. Moreover, for all values $r>r^\sharp$ there exist two distinct values $\gamma^1_r,\gamma^2_r \in[0,1)$ so that, for all $i\in I$, the profiles of play

\begin{align*}
    \sigma_{i,1}^*(C\mid \zeta) =
\begin{cases}
1, & \zeta \in\{(0,0), (1,n)\}\\
\gamma^1_r, & \text{otherwise},
\end{cases}
\end{align*}

\begin{align*}
    \sigma_{i,2}^*(C\mid \zeta) =
\begin{cases}
1, & \zeta \in\{(0,0), (1,n)\}\\
\gamma^2_r, & \text{otherwise}
\end{cases}
\end{align*}
 establish two distinct sequential equilibria $(\sigma^*_1, \mu^*_1)$ and $(\sigma^*_1, \mu^*_2)$, respectively.

\end{conjecture}
A positive answer to the above would also extend the results from \textcite{gallice2019co} by providing a \textit{mixed strategy equilibrium} to their results for $m>1$, and we postpone this to future work.

\newpage
\printbibliography

\newpage
\section{Appendix: Proofs}\label{apendix:Proofs}
Let's assume that we have $b = \frac{N}{n}$ many groups with $n$ individuals per group. Assume that the standard strategy given profile $\zeta$ is
\begin{align*}
    \sigma^*(C\mid \zeta) =
\begin{cases}
1, & \zeta \in \{(0,0), (1,n)\}\\
\gamma, & \text{otherwise}.
\end{cases}
\end{align*}
where $\gamma \in[0,1)$. Fix a profile $\overline{\zeta}$ and set for Player $j$ the strategies
\[
\sigma_j^C(\zeta)
=
\begin{cases}
\sigma^*_j(\zeta), & \zeta \not = \overline{\zeta}\\
1, & \zeta = \overline{\zeta}.
\end{cases}
\]

\[
\sigma_j^D(\zeta)
=
\begin{cases}
\sigma^*_j(\zeta), & \zeta \not = \overline{\zeta}\\
0, & \zeta = \overline{\zeta},
\end{cases}
\]
and $\mu_j^D$ and $\mu_j^C$ as their corresponding beliefs. Set $\phi_t(\gamma)$ and $\psi_t(\gamma)$ as the number of additional contributions said player expects from contributing rather than defecting whilst in Group $t$, and the likelihood  of being in Group $t$ after observing a defection, respectively. For $n>1$, $\phi_t(\gamma)$ is different depending on the history $\overline{\zeta} = (1, n')$ observed by Player $j$. Indeed, if $\overline{\zeta} = (1, n')$ with $n>n'$ then all other member of Group $t$' will also observe a defection and act according to their strategies. Hence, there is a non-zero probability that subsequent groups will also witness a sample with defection even if Player $j$ itself contributes. In contrast, if $\overline{\zeta} = (1, n)$ then Player $j$'s defection will be the only one in her group and the effect of her defection should be larger than the previous scenario. The following lemma demonstrates just that.

\begin{lemma}\label{lem:psiPhi} Both $\phi_t(\gamma), \psi_t(\gamma): [0,1] \to \R$ are continuous functions where:
\begin{enumerate}
\item given $\overline{\zeta} = (1, n')$ with $n>n'$
\[ \phi_t(\gamma) = \frac{n(1-\gamma)(1-(1-\gamma^n)^{b-t})}{\gamma}+1.
\]
with $\phi_t(0) = 1$ for all $n>1$ and $\phi_t(0) = b-t+1$ for $n=1$; 
\item given $\overline{\zeta} = (1, n)$

\[ \phi_t(\gamma) = \frac{n(1-\gamma)(1-(1-\gamma^n)^{b-t})}{\gamma^n}+1.
\]
with $\phi_t(0) = (b-t)n +1$; and
\item  \[
\psi_t(\gamma) = \frac{(1-(1-\gamma^n)^{t-1})}{b-1-\gamma^{-n}(1-\gamma^n)(1-(1-\gamma^n)^{b-1})}
\]
with $\psi_t(0)  =\frac{2(t-1)}{b(b-1)}$.
\end{enumerate}
\end{lemma}

\begin{proof}

\noindent
The number of additional contributors a player in Group $t$ should expect from contributing rather than defecting given history $\overline{\zeta} = (1,n')$ becomes 
\[
\phi_t(\gamma) = E_{\mu^C}(G_{-j}\mid \zeta = \overline{\zeta}, Q(j) = t) - E_{\mu^D}(G_{-j}\mid \zeta = \overline{\zeta}, Q(j) = t).
\]
where
\begin{align*}
E_{\mu^C}(G_{-j}\mid \zeta = \overline{\zeta}, Q(j) = t)  &= \sum_{i=1}^{t-1}E_{\mu^*}(G_{i}\mid \zeta = \overline{\zeta}, Q(j) = t)\\
& +  \sum_{i=t}^{s} E_{\mu^C}(G_{i}\mid \zeta = \overline{\zeta}, Q(j) = t),
\end{align*}
\begin{align*}
E_{\mu^D}(G_{-j}\mid \zeta = \overline{\zeta}, Q(j) = t)  &= \sum_{i=1}^{t-1}E_{\mu^*}(G_{i}\mid \zeta = \overline{\zeta}, Q(j) = t)\\
& +  \sum_{i=t}^{s} E_{\mu^D}(G_{i}\mid \zeta = \overline{\zeta}, Q(j) = t),
\end{align*}
and $G_i$ represents the $i^\text{th}$ group.
Hence,
\[
\phi_t(\gamma) = \sum_{i=t}^{b} E_{\mu^C}(G_{i}\mid \zeta = \overline{\zeta}, Q(j) = t) -   E_{\mu^D}(G_{i}\mid \zeta = \overline{\zeta}, Q(j) = t).
\]
\noindent
1. If the sample $\overline{\zeta} = (1,n')$ contains a defection (i.e., $n>n'$) we have
\[
E_{\mu^D}(G_{t}\mid \zeta = \overline{\zeta}, Q(j) = t) =   (n-1)\gamma 
\]
and
\[
E_{\mu^C}(G_{i}\mid \zeta = \overline{\zeta}, Q(j) = t) = \sum_{i=t}^{b} E_{\mu^D}(G_{i}\mid \zeta = \overline{\zeta}, Q(j) = t) + 1.
\]
\noindent
For $t+1$ we get 
\[
E_{\mu^D}(G_{t+1}\mid \zeta = \overline{\zeta}, Q(j) = t)= \gamma n 
\]
\text{ and } 
\[
E_{\mu^C}(G_{t+1}\mid \zeta = \overline{\zeta}, Q(j) = t) = \gamma^{n-1}n + (1-\gamma^{n-1})\gamma n.
\]
In general, for $t+k$ with $k\geq 1$ we have
\[
E_{\mu^D}(G_{t+k}\mid \zeta = \overline{\zeta}, Q(j) = t) = n\left[1-(1-\gamma^n)^{k-1}(1-\gamma)\right] 
\]
and 

\[
E_{\mu^C}(G_{t+k}\mid \zeta = \overline{\zeta}, Q(j) = t) = n\left[\gamma^{n-1}+ (1-\gamma^{n-1}) (1-(1-\gamma)(1-\gamma^n)^{k-1}\right] 
\]
In turn, 
\[
E_{\mu^C}(G_{t+k}\mid \zeta = \overline{\zeta}, Q(j) = t) - E_{\mu^D}(G_{t+k}\mid \zeta = \overline{\zeta}, Q(j) = t) = n\gamma^{n-1}(1-\gamma)(1-\gamma^n)^{k-1}.
\]
As a sum of powers of $(1-\gamma^n)^{k-1}$ we deduce tha for any $\gamma \in (0,1]$:
\[
\phi_t(\gamma) = \frac{n(1-\gamma)(1-(1-\gamma^n)^{b-t})}{\gamma}+1.
\]
with
\[
\lim_{\gamma\to 0} \phi_t(\gamma) = \lim_{\gamma\to 1} \phi_t(\gamma) = 1.
\]
for all $n>1$. It is interesting to observe that $\phi_t(\gamma)>1$ for all $
\gamma \in (0,1)$. In fact, the distribution $\phi_t(\gamma)$ is bell shaped. This means that there is an optimal value of $\gamma$ that will maximise the additional contributions aplayer can expect by contributing rather than defecting. If we let $n=1$ we get 
\[
\phi_t(\gamma) =\frac{1-(1-\gamma)^{b-t+1}}{\gamma}
\]
with 
\[
\lim_{\gamma\to 0} \phi_t(\gamma) = b-t+1.
\]
This is precisely what is obtained in \textcite{gallice2019co}.\\

2. If sample $\overline{\zeta} = (1,n)$ then the computation is similar but much simpler. For $t+k$ with $k\geq 1$ we have
\[
E_{\mu^D}(G_{t+k}\mid \zeta = \overline{\zeta}, Q(j) = t) = n\left[1-(1-\gamma^n)^{k-1}(1-\gamma)\right] 
\]
and 

\[
E_{\mu^C}(G_{t+k}\mid \zeta = \overline{\zeta}, Q(j) = t) = n 
\]
In turn, 
\[
E_{\mu^C}(G_{t+k}\mid \zeta = \overline{\zeta}, Q(j) = t) - E_{\mu^D}(G_{t+k}\mid \zeta = \overline{\zeta}, Q(j) = t) = n(1-\gamma)(1-\gamma^n)^{k-1}.
\]
Therefore, for any $\gamma \in (0,1]$ we have
\[
\phi_t(\gamma) = \frac{n(1-\gamma)(1-(1-\gamma^n)^{b-t})}{\gamma^n}+1.
\]
with
\[
\lim_{\gamma\to 0} \phi_t(\gamma) = (b-t)n+1.
\]

\noindent

3. Let $\psi_t(\gamma)$ denote the likelihood that a player finds itself in position $t$ after witnessing a defection. It follows that
\begin{align*}
\psi_t(\gamma) = \frac{\sum_{j=1}^{t-1}(1-\gamma^n)^j}{\sum_{k=2}^b\sum_{i=1}^{k-1}(1-\gamma^n)^{t-i-1}} &= \frac{ \gamma^{-n}(1-(1-\gamma^n)^{t-1})}{\gamma^{-n}\left(b-1+\gamma^{-n}(1-\gamma^n)(1-(1-\gamma^n)^{b-1}\right)}\\
&= \frac{(1-(1-\gamma^n)^{t-1})}{b-1-\gamma^{-n}(1-\gamma^n)(1-(1-\gamma^n)^{b-1})}
\end{align*}

For all other $n>1$, we can make the replacement $y = \gamma^p$ to obtain after, applying L'Hospital's Rule, $\psi_t(0) =\frac{2(t-1)}{b(b-1)}$.\\
\end{proof}

\noindent
{\bf Proofs of Theorems~\ref{thm:pureStrat} and~\ref{thm:mixedStrat}.}\\ \label{Proof:Thms1and2}

In what follows it becomes useful to let $\phi_t(\gamma, \zeta)$ denote $\phi_t(\gamma)$ given a profile $\zeta$. A player in Group $1$ contributes whenever $\frac{r}{N}\phi_1(\gamma, (0,0)) - 1 \geq 0$; a player who received a sample $\overline{\zeta} = (1,n)$ contributes provided $\sum_{t=2}^b\frac{1}{b-1}\phi_t(\gamma, \overline{\zeta})-1\geq 0$; and given profile $\overline{\zeta}' = (1,n')$ with $n'<n$ a player contributes provided $\sum_{t=2}^b \psi_t(\gamma)\phi_t(\gamma, \overline{\zeta}')-1\geq 0$.

\begin{lemma}\label{lem:inequalities} Given profiles $\overline{\zeta}' = (1,n')$ with $n'<n$ and $\overline{\zeta} = (1,n)$ it follows that for all $\gamma \in [0,1]$

\[
\phi_1(\gamma) > \sum_{t=2}^b\frac{1}{b-1}\phi_t(\gamma, \overline{\zeta}) \geq \sum_{t=2}^b \psi_t(\gamma)\phi_t(\gamma,\overline{\zeta}').
\]
\end{lemma}
\begin{proof} The first inequality is obvious. For the second one observe that by Lemma~\ref{lem:psiPhi}, $\phi_t(\gamma,\overline{\zeta}') \leq \phi_t(\gamma,\overline{\zeta})$ for all $t$. In turn, showing that 
\[
\sum_{t=2}^b\frac{1}{b-1}\phi_t(\gamma, \overline{\zeta}') \geq \sum_{t=2}^b \psi_t(\gamma)\phi_t(\gamma,\overline{\zeta}')
\]
suffices. Fix a $\gamma \in [0,1]$ and observe that since
\[
1= \sum_{t=2}^b\frac{1}{b-1} = \sum_{t=2}^b\psi_t(\gamma)
\]
and $\psi_2(\gamma) < \frac{1}{b-1}$ then there must exist $t^*\leq b$ with $\psi_{t}(\gamma) > \frac{1}{b-1}$ for all $t>t^*$. Consequently, 
\[
 \sum_{t=2}^{t^*}\frac{1}{b-1} \geq \sum_{t=2}^{t^*}\psi_t(\gamma) \text{ and }  \sum_{t=t^*+1}^{b}\frac{1}{b-1} \leq \sum_{t=t^*+1}^{b}\psi_t(\gamma).
\]
Since $\phi_t(\gamma)$ is decreasing in $t$ the claim follows.
\end{proof}

Next, setting \[
H(\gamma) = \frac{r}{N}\sum_{t=2}^b\psi_t(\gamma)\phi_t(\gamma) - 1
\]
we find that 
\[
H(0) =\frac{r}{N} \sum_{t=2}^n\psi_t(0)\phi_t(0) - 1 = \frac{r}{N} - 1 <0
\]
for all $n>1$.  Thus, a pure contribution strategy exists for all reasonable values of $r$ (i.e., $r<N$) and Theorem~\ref{thm:pureStrat} is proved. \\

In terms of Theorem~\ref{thm:mixedStrat} there are plenty of values of $r$ that yield a $\gamma$ with $H(\gamma) = 0$. Since $n>1$ observe that $H(\gamma) > H(1) = H(0) = \frac{r}{N} -1$ for all $\gamma \in (0,1)$. In turn, by Rolle's Theorem there exists at least one local maxima for $H(\gamma)$ in $(0,1)$. Of course, said maxima must have the same $\gamma$ value for all $r$ as the latter is a constant factor on $H(\gamma)$. Set $\gamma^\sharp$ to denote this maxima. Consider $H(\gamma)$ as a function on $r$ and $\gamma$, $H(r,\gamma)$. Since for all $\gamma \in (0,1) $ we have $H(N,\gamma)>0$ and $H(0,\gamma)<0$, and $H(r,\gamma^\sharp)$ is continuous on $r$ it follows that there exists a unique value $r^\sharp$ with $H(r^\sharp, \gamma^\sharp) = 0$. Moreover, for all $r>r^\sharp$ we get $H(r,\gamma^\sharp) > 0$ and, thus, Theorem~\ref{thm:mixedStrat} is proved. \hfill \qedsymbol{}\\

\noindent
{\bf Proof of Theorem~\ref{thm:mixedStratAsym}.}\\

 As before, let's assume that the standard strategy given profile $\zeta$ is
\begin{align*}
    \sigma^*(C\mid \zeta) =
\begin{cases}
1, & \zeta \in \{(0,0), (1,n_k)\}\\
\gamma, & \text{otherwise},
\end{cases}
\end{align*}
where $1\leq k < b$. Observe a player in Group $t>1$ must be aware of the size of Group $t-1$ in order to be able to detect a defection. Fix a profile $\overline{\zeta}$ and set for Player $j$ the strategies
\[
\sigma_j^C(\zeta)
=
\begin{cases}
\sigma^*_j(\zeta), & \zeta \not = \overline{\zeta}\\
1, & \zeta = \overline{\zeta}.
\end{cases}
\]

\[
\sigma_j^D(\zeta)
=
\begin{cases}
\sigma^*_j(\zeta), & \zeta \not = \overline{\zeta}\\
0, & \zeta = \overline{\zeta},
\end{cases}
\]
and $\mu^D$ and $\mu^C$ as their corresponding beliefs. Let $\phi_t(\gamma)$ denote the expected additional contribution from contributing rather than defecting provided Player $j$ belongs to Group $t$. As with the symmetric case
\[
\phi_t(\gamma) = \sum_{i=t}^{b} E_{\mu^C}(G_{i}\mid \zeta = \overline{\zeta}, Q(j) = t) -   E_{\mu^D}(G_{i}\mid \zeta = \overline{\zeta}, Q(j) = t).
\]
As with Lemma~\ref{lem:psiPhi}, the form of $\phi_t(\gamma)$ depends on the type of sample Player $j$ receives. Set $\Delta_i(\overline{\zeta})= E_{\mu^C}(G_{i}\mid \zeta = \overline{\zeta}, Q(j) = t) -   E_{\mu^D}(G_{i}\mid \zeta = \overline{\zeta}, Q(j)= t)$. If the sample $\overline{\zeta} = (1,n')$ contains a defection (i.e., $n_{t-1}>n'$) we have that
\[
 \Delta_i(\overline{\zeta}) = n_i(1-\gamma)\gamma^{n_t-1}\left(1-\sum_{l=1}^{k}\gamma^{n_l}\prod_{w=1}^{i-1}(1-\gamma^{n_w})\right).
\]
Similarly, if $\overline{\zeta}= (1,n_{t-1})$ then 
\[
 \Delta_i(\overline{\zeta}) = n_i(1-\gamma)\left(1-\sum_{l=1}^{k}\gamma^{n_l}\prod_{w=1}^{i-1}(1-\gamma^{n_w})\right).
\]

Observe that any time $n_t = 1$ both of the above expressions coincide, as it is to be expected. Evidently, $\Delta_i(\overline{\zeta})$ is larger when $\overline{\zeta}$ contains no defections. Computing $\phi_t(\gamma) = \sum_{i=t}^{b} \Delta_i(\overline{\zeta})$ presents an onerous task regardless of the sample $\overline{\zeta}$. What we do instead is bound each $\Delta_i(\overline{\zeta})$ between two computationally simpler functions. As the following lemma suggests, it suffices to focus on doing so for $\phi(\gamma)$ when $\overline{\zeta}$ contains a defection. As before, let $\psi_t(\gamma)$ represent the likelihood  of being in Group $t$ after observing a defection and $\phi_t(\gamma, \overline{\zeta})$ denote $\phi_t(\gamma)$ on sample $\overline{\zeta}$.

\begin{lemma} Given profiles $\overline{\zeta}' = (1,n')$ with $n'<n_{t-1}$ and $\overline{\zeta} = (1,n_{t-1})$ it follows that for all $\gamma \in [0,1]$

\[
\phi_1(\gamma) > \sum_{t=2}^b\frac{1}{b-1}\phi_t(\gamma, \overline{\zeta}) \geq \sum_{t=2}^b \psi_t(\gamma)\phi_t(\gamma,\overline{\zeta}').
\]
\end{lemma}
\begin{proof} The proof of this Lemma is almost identical to that of Lemma~\ref{lem:inequalities}.
\end{proof}
Let $M  =\max\{n_1,\ldots, n_b\}$ and $\lambda = \min\{n_1,\ldots, n_b\}$. One can easily verify that
\begin{align*}
     \lambda(1-\gamma)\gamma^{M-1}\left(1-\sum_{i=1}^{k}\gamma^{\lambda}(1-\gamma^{M})^i\right)&\leq \Delta_i \\
 &\leq M(1-\gamma)\gamma^{\lambda-1}\left(1-\sum_{i=1}^{k}\gamma^{M}(1-\gamma^{\lambda})^i\right)
\end{align*}
In fact, $M = \lambda$ describes the symmetric scenario. Setting
\[
\phi_t^{\top}(\gamma) = \sum_{i=t+1}^b \lambda(1-\gamma)\gamma^{M-1}\left(1-\sum_{i=1}^{k}\gamma^{\lambda}(1-\gamma^{M})^i\right) + 1
\]
and
\[\phi_t^{\bot}(\gamma) =\sum_{i=t+1}^b \lambda(1-\gamma)\gamma^{\lambda-1}\left(1-\sum_{i=1}^{k}\gamma^{M}(1-\gamma^{\lambda})^i\right) + 1
\]
we get that 
\[
\phi_t^{\top}(\gamma) = \frac{M\left(1-\gamma\right)}{\gamma}\left(\gamma^{-1+\lambda}(b-t-1)\left(1-\gamma^{-1+M}\right)+\gamma^{-1-\lambda+M}\left(1-\left(1-\gamma^{\lambda}\right)^{-1-t+b}\right)\right)+1
\] and
\[
\phi_t^{\bot}(\gamma) = \frac{\lambda\left(1-\gamma\right)}{\gamma}\left(\gamma^{-1+M}(b-t-1)\left(1-\gamma^{-1+\lambda}\right)+\gamma^{-1-M+\lambda}\left(1-\left(1-\gamma^{M}\right)^{-1-t+b}\right)\right)+1
\]
 with 
 \[
 \phi_t^{\bot}(\gamma) \leq \phi_t(\gamma) \leq \phi_t^{\top}(\gamma).
 \]
 Moreover, that
 \[
 \lim_{\gamma \to 0} \phi_t^{\bot}(\gamma) = \lim_{\gamma \to 1} \phi_t^{\bot}(\gamma) = \lim_{\gamma \to 0} \phi_t^{\top}(\gamma) = \lim_{\gamma \to 1} \phi_t^{\top}(\gamma) = 1
 \]
 forces
 \[
 \lim_{\gamma \to 0} \phi_t(\gamma) = \lim_{\gamma \to 1} \phi_t(\gamma) = 1
 \]
 by a pinching $\phi_t(\gamma)$ between its bounds.  Setting  
 \[
 H(\gamma) = \frac{r}{N}\sum_{t=2}^b\psi_t(\gamma)\phi_t(\gamma) - 1
 \]
 we get that for all $\gamma \in (0,1)$: $H(\gamma) > H(1)  = H(0) =  \frac{r}{N} -1<1$ and for $r=N$, $H(\gamma) > 0$. In turn, we can apply the same continuity arguments employed in the proofs of Theorems~\ref{thm:pureStrat} and~\ref{thm:mixedStrat} to derive the desired result. \hfill \qedsymbol{}

\end{document}